\documentclass[a4paper,12pt]{article}
\usepackage{latexsym}
\usepackage{amsmath,amssymb,mathrsfs,pictex} 
\usepackage{amsmath,amsthm,amsfonts,amscd,eucal,bm}
\usepackage{graphicx}
\usepackage{epsfig}
\usepackage{epstopdf} 
\usepackage{cancel}

\numberwithin{equation}{section}

\def\cf{{\mathcal F}}

\def\ch{{\mathcal H}}

\def\ck{{\mathcal K}}

\def\cs{{\mathcal S}}

\def\bc{{\mathbb C}}

\def\bn{{\mathbb N}}

\def\br{{\mathbb R}}
\def\bs{{\mathbb S}}
\def\bt{{\mathbb T}}

\def\bz{{\mathbb Z}}

\def\a{\alpha}
\def\b{\beta}
\def\g{\gamma}        \def\G{\Gamma}
\def\d{\delta}        
\def\eps{\varepsilon} 
     
\def\z{\zeta}

       \def\La{\Lambda}
\def\m{\mu}
\def\n{\nu}

\def\r{\rho}
\def\s{\sigma}

        \def\Om{\Omega}

\newtheorem{Thm}{Theorem}[section]

\newtheorem{Prop}[Thm]{Proposition}

\theoremstyle{definition}

\theoremstyle{remark}

\def\di{\mathrm{d}}

\newcommand{\tr}{\textrm{Tr}}

\newcommand{\ty}[1]{\mathop{\rm {#1}}}

\newcommand{\rmd}{\textrm{d}}

\title{Spectral actions for $q$-particles\\and their asymptotics}

\author{{\sc Fabio Ciolli, Francesco Fidaleo}\\
{ciolli@mat.uniroma2.it, fidaleo@mat.uniroma2.it}
\\
Dipartimento di Matematica,
Universit\`a di Roma Tor Vergata,\\
Via della Ricerca Scientifica, 1, I-00133 Roma, Italy
}

\date{}
\begin{document}

\maketitle

\begin{abstract}
For spectral actions consisting of the average number of particles and arising from open systems made of general free $q$-particles (including Bose, Fermi and classical ones corresponding to $q=\pm1$ and $0$, respectively) in thermal equilibrium, we compute the asymptotic expansion with respect to the natural cut-off. We treat both relevant situations relative to massless and non relativistic massive particles, where the natural cut-off is $1/\b=k_{\rm B}T$ and $1/\sqrt{\b}$, respectively.  
We show that the massless situation enjoys less regularity properties than the massive one. We also treat in some detail the relativistic massive case for which the natural cut-off is again $1/\b$.
We then consider the passage to the continuum describing infinitely extended open systems in thermal equilibrium, by also discussing the appearance of condensation phenomena occurring for Bose-like $q$-particles, $q\in(0,1]$. We then compare the arising results for the finite volume situation (discrete spectrum) with the corresponding infinite volume one (continuous spectrum).
\end{abstract}
\smallskip\small{
}

\section{Introduction}
Very recently, the investigation of the fascinating topic called Connes' {\it noncommutative geometry} had a formidable growth for many potential applications to mathematics and physics. Among those, we mention the possible role in the attempt to solve the, still open, long-standing problem of the Riemann conjecture involving the zeroes of the zeta function (e.g.\ \cite{C1}). We also mention the crucial problem concerning the attempt to provide a model, still not available, which unify the gravitation with the three remaining elementary interactions (e.g.\ \cite{CC}). The reader is referred to the monograph \cite{Co} for a quite complete and understandable treatise on noncommutative geometry.

The main ingredient arising in noncommutative geometry is the so-called {\it Dirac operator}, which is a kind of square root of the Laplace operator. The Dirac operator is supposed to encode the main properties of the associated, commutative or non commutative, manifold under consideration.
The Dirac operator enters in the definition of the so-called {\it spectral triples} (e.g.\ \cite{Co, CM} and the references cited therein), and the {\it spectral action}
(e.g.\ \cite{CC}). 

Concerning the spectral triples, most of the investigation is devoted to commutative manifolds, and noncommutative ones equipped with a canonical tracial state 
(i.e.\ the easiest noncommutative generalisation of the Lebesgue/Liouville measure), and thus providing von Neumann factors of type ${\rm II}_1$. Only recently in 
\cite{FS, FH, CF1}, it was considered the possibility to investigate examples based on the noncommutative 2-tori, for which the involved Dirac operator is suitably deformed by the use of the Tomita modular operator, the last being unbounded and not ``inner" in all interesting cases. Therefore, such deformed spectral triples, called modular since the modular data might play a crucial role, naturally arise from type ${\rm III}$ representations.

The spectral action $S$ was introduced in \cite{CC} with the aim to provide a set of consistent equations in which all fundamental interactions in nature are unified. Also in this case, the Dirac operator $D$ plays a crucial role because the spectral action is supposed to have the form $S(D,\La)=\tr\big(f(|D|/\La)\big)$. Here, the function $f$ is associated to an extensive quantity, and $\La$ is a natural cut-off determined by the properties of the underlying physical system. It should be also noted that the asymptotic expansion of the heat kernel associated to the (powers of the) Dirac operator $D$ and the spectral action $S(D,\La)$ with respect to the cut-off $\La$ for 
$\La\uparrow\infty$, is an important tool in most approaches to spectral geometry. For further details on this crucial point, the reader is referred to \cite{CC, CC1, E, EGV, EI} and the references cited therein.

Very recently, in the paper \cite{CCS} it was investigated the spectral action associated to the extensive quantity, as the entropy, of some quite general quantum systems describing Fermi particles or, in other words, systems based on the Canonical Anti-commutation Relations. In \cite{DKS}, the investigation of the asymptotics of some spectral actions has been extended 
to particular Hamiltonians, also involving the Bose case. 

By following the mentioned  paper \cite{DKS}, it is then natural to address the investigation of the asymptotics of the spectral actions associated to infinitely extended systems describing Bose/Fermi free gas, and also classical particles obeying to the Boltzmann statistics. Since the properties of such free gases are now available for the general cases $q\in[-1,1]$, relative to $q$-particles or quons, after computing their {\it grand partition function} (cf.\ \cite{CF}), we carry out such an investigation for all cases, $q=\pm1$ and $q=0$ being the Bose/Fermi and Boltzmann cases, respectively.

Indeed, for spectral actions arising from open systems made of such general free $q$-particles in thermal equilibrium, we compute the asymptotic expansion with respect to the natural cut-off associated to $1/\b=k_{\rm B}T$, $T$ being the absolute temperature and $k_{\rm B}\approx 1.3806488\times10^{-23}\,
J K^{-1}$ the Boltzmann constant. For the sake of simplicity, we deal with the spectral actions associated to the ``average number of particles", by considering both relevant situations for non relativistic massive, and massless cases. The spectral actions associated to other 
extensive quantities, like average energy and entropy, can be analogously studied.

Actually, we show that the massless situation enjoy less regularity properties than the massive one. We also investigate in some detail the relativistic massive case, which appears more involved than the corresponding non relativistic one. 

We also consider the passage to the continuum describing infinitely extended open systems in thermal equilibrium for which the Hamiltonian, being a suitable function of the Dirac operator, has continuous spectrum. 

In the context of the continuum, we briefly discuss the appearance of condensation phenomena, occurring for Bose-like $q$-particles where $q\in(0,1]$, by pointing out that the condensation provides an additional addendum in the asymptotics. We also compare the results relative to the asymptotics for the finite volume situation (discrete spectrum) with infinite volume one (continuous spectrum).

We would also like to mention the so-called {\it anyons}, which are particles satisfying, at least formally, a commutation relation similar to \eqref{qcra} enjoyed by quons, where now $q$ is any value in the unit circle $\bt\equiv\{z\in\bc\mid|z|=1\}$.  We point out that such particles might have a role in the, still open problem of the explanation of the integer and fractional {\it quantum Hall effect}, 
see e.g.\ \cite{ZZ}.

At this stage, the $q$-deformed particles ($q\in[-1,1]$) seem to be related to quantum groups and quantum algebras, which drew much attention decades ago. In fact, such particles naturally emerge from exactly solvable models in statistical mechanics which acquire the Yang-Baxter equation. We also point out that the irreducible representations of $q$-deformed particles are substantial extensions of the quantum algebra in connections to the braid group statistics. For such interesting applications of these quons, the reader is referred to the monograph \cite{ZZZ}.

The standard notations and basic results relative to the topics involved in the present note are borrowed by the current literature, partially reported at the end, to which the reader is referred.

\section{The grand-partition function for $q$-particles, $q\in[-1,1]$}
\label{1mic1}

We start with a system whose Hamiltonian $H$ is a selfadjoint positive (i.e.\ $\s(H)\subset[0,+\infty)$) operator with compact resolvent, acting on a separable Hilbert space $\ch$, called the {\it one-particle space}.

In such a situation, the spectrum $\s(H)$ is made by isolated points, accumulating at $+\infty$ if $\ch$ is infinite dimensional. In addition, the multiplicity 
$g(\eps)$ of each eigenvalue $\eps\in\s(H)$ is finite. Summarising, by considering the resolution of the identity of $H$, we have 
$I_\ch=\sum_{\eps\in\s(H)}P_{\eps}$, 
$$
H=\sum_{\eps\in\s(H)}\eps P_{\eps},\,\,\text{and}\,\,g(\eps):=\text{dim}\big(\text{Ran}(P_{\eps})\big)<\infty\,.
$$

We also suppose that at any inverse temperature $\b:=1/k_B T$, $e^{-\b H}$ is trace-class for each $\b>0$, and define the {\it partition function}
$\z:=\ty{Tr}\!\big(e^{-\b H}\big)$. 

In the present paper, we generally deal with the so-called $q$-{\it particles}, usually named as {\it quons}, $q\in[-1,1]$. Such exotic $q$-particles are naturally associated to the following commutation relations (e.g.\ \cite{JSW, BKS})
\begin{equation}
\label{qcra}
{\bf a}_q(f){\bf a}_q^\dagger(g)-q{\bf a}_q^\dagger(g){\bf a}(f)=\langle g,f\rangle I_\ch\,, \quad f,g\in\ch\,,
\end{equation}
$\ch$ being the one-particle space, enjoyed by the creators and annihilators. Such commutation relations can be viewed as an interpolation between that associated to particles obeying to the Fermi statistics (i.e.\ $q=-1$) and that of particles obeying to the Bose statistics (i.e.\ $q=1$), passing for the value $q=0$ describing the classical particles, and so obeying to the Boltzmann statistics. They can be realised by (unbounded in the case $q=1$) operators acting on the corresponding Fock spaces $\cf_q$, see e.g.\ \cite{BR, BKS}.

Concerning the grand partition function ${Z}$, it comes by considering open systems in thermodynamic equilibrium at inverse temperature $\b$ and 
{\it chemical potential} $\m$. It is customary to express the grand partition function ${Z}={Z}(\b,z)$ in terms of the independent variables, which are the inverse temperature $\b$ and the so-called {\it activity} $z:=e^{\b\m}$. When it causes no confusion, we omit to indicate such dependences. We would like to point out that, since the activity $z$ should be considered as an independent thermodynamic variable, it is not involved in the computation of the expansion of the  quantities w.r.t.\ $\b\downarrow0$.
Since we deal with the general cases $q\in[-1,1]$, the partition function
${Z}=Z_q={Z}_q(\b,z)$ is also a function of such an additional parameter $q$. 

The partition function is computed as $\tr\left(e^{-\b K}\right)$
by using the second quantisation grand canonical Hamiltonians $K:=\rmd\G(H)-\m N$ acting on the corresponding Fock space. Here, 
$N$ is the {\it number operator}, and the Fock space would depend on the statistics to which the involved $q$-particles obey. 

Unfortunately, it is first seen in \cite{W} and then in \cite{CF}, that such a method to compute the grand partition function works only for the Bose/Fermi cases corresponding to $q=\pm1$, see e.g.\ \cite{BR}. Such a computation fails even for the classical (or Boltzmann) case $q=0$, where the statistics is however well-known (e.g.\ \cite{FV}, Section 2), but it is necessary to correct by the Gibbs factor as follows:
$$
\sum_{n=1}^{+\infty}\frac{\big(\tr\,e^{-\b H}\big)^n}{n!}z^n,\,\,\,\text{instead of}\,\,\,\sum_{n=1}^{+\infty}\big(\tr\,e^{-\b H}\big)^nz^n\,.
$$

Concerning the remaining cases $q\in(-1,0)\bigcup(0,1)$, the partition function should be computed by a different method because it is completely unknown what should be the statistics enjoyed by such exotic particles. However, in \cite{CF} it was computed the reasonable partition function for all $q\in[-1,1]$, uniquely determined up to a multiplicative constant.

Indeed, by using only the commutation relations \eqref{qcra} and imposing the Kubo-Martin-Schwinger condition in the grand canonical ensemble scheme, in \cite{AF5} it was proven for such a free gas of quons, that the average population of particles for a fixed energy-level $\eps$ in thermodynamic equilibrium satisfies the natural extension 
\begin{equation}
\label{pldi}
\overline{\n}_q(\eps)=\frac{1}{z^{-1}e^{\b\eps}-q}\,,\quad -1\leq q\leq1\,,
\end{equation}
of the well-known Planck distribution $\frac{1}{z^{-1}e^{\b\eps}-1}$.

By taking in account the degeneracy $g(\eps)$ of the single level $\eps$, we get the average population
\begin{equation}
\label{pldibis}
n_q(\eps)=g(\eps)\overline{\n}_q(\eps)=\frac{g(\eps)}{z^{-1}e^{\b\eps}-q}\,,\quad -1\leq q\leq1\,.
\end{equation}
For particles living in the continuum $\br^n$ as for the analysis in \cite{AF5}, the degeneracy $g(\eps)$ is automatically absorbed, see \eqref{b4}.

On the other hand, the same distribution \eqref{pldi} was recovered in \cite{FV}, in the scheme of microcanonical ensemble, by maximising the $q$-{\it entropy functional} 
\begin{equation}
\label{0entc}
S_q(\{\n(\eps)\}):=k_{\rm B}\sum_{\eps\in\s(H)}g(\eps)\bigg[\frac{\big(1+q\n(\eps)\big)}{q}\ln\big(1+q\n(\eps)\big)-\n(\eps)\ln\n(\eps)\bigg]\,,
\end{equation}
w.r.t.\ the set of variables $\{\n(\eps)\}_{\eps\in\s(H)}$ under the constrains 
\begin{equation}
\label{plq1}
\sum_{\eps\in\s(H)}\eps g(\eps)\n(\eps)=E\,, \quad \sum_{\eps\in\s(H)}g(\eps)\n(\eps)=N\,.
\end{equation}
Here, $E$ and $N$ are the pre-assigned values of the total energy and the total number of particles of the physical system under consideration.

To simplify, in \cite{FV} we have supposed that $|\s(H)|<+\infty$. A more refined analysis can be carried out also in the case of infinite spectrum, obtaining the same result.
In addition, it was noticed that \eqref{0entc} reduces to
\begin{equation}
\label{000entc}
S_0(\{\n(\eps)\}):=k_{\rm B}\sum_{\eps\in\s(H)}g(\eps)\n(\eps)\big(1-\ln\n(\eps)\big)
\end{equation}
in the Boltzmann case $q=0$.

Consequently, the mean entropy of a system in thermal equilibrium will be
\begin{equation}
\label{entc9}
S_q=S_q(\b,z)=S_q(\{\n(\eps)\})\big|_{\n(\eps)=\overline{\n}_q(\eps)}\,,
\end{equation}
with the $\overline{\n}(\eps)$ given in \eqref{pldi}.

\medskip

By coming back to the grand partition functions, from \cite{CF} we recall that
\begin{equation}
\label{occq3}
\begin{split}
Z_q=e^{\displaystyle -\frac{\displaystyle \tr\ln\left(I-zqe^{\displaystyle -\b H}\right)}{\displaystyle q}}\,,&\quad 
\left\{\begin{array}{ll}
0<z<\frac{\displaystyle e^{\displaystyle \b\min\s(H)}}{\displaystyle q}\,\,& 0<q<1\,, \\
{}&\\
z>0\,\,& -1<q<0\,.
\end{array}\right.\\
Z_0=e^{z\tr\big(e^{-\b H}\big)}\,,&\quad z>0\,.
\end{split}
\end{equation}

The appropriate form in \eqref{occq3}
of the grand partition function is indeed computed according the prescription given in the following
\begin{Thm}[\cite{CF}, Prop. 3]
For each $\eps\in\s(H)$ fixed, $Z_q$, $q\in[-1,1]$, given in \eqref{occq3} is derivable w.r.t.\ $\eps$ and we have
$$
-\frac1{\b}\frac{\partial\ln Z_q}{\partial\eps}=n_q(\eps)\,,
$$
where the $n_q(\eps)$ are the occupation numbers in \eqref{pldi}. 
\end{Thm}

We also note that $\lim_{q\to0}Z_q=Z_0$, uniformly in the thermodynamic parameters $\b,z$, as stated in Proposition 4 of \cite{CF}.

\section{Spectral actions for $q$-particles, $q\in[-1,1]$}
\label{2mic1}

We discuss the preliminary aspects of some natural spectral actions for general $q$-particles, by taking into consideration the concrete cases relative to non relativistic massive, and massless free particles. The case of relativistic massive particles is treated in some details in Section \ref{rmsex}.

The next sections will be devoted to the detailed expansion involving the infinite volume theory after the passage to the continuum for which the Hamiltonian exhibits continuum spectrum, as well as finite volume theories whose spectra are discrete.
\medskip

We therefore start with the Hamiltonian which is a suitable function of the Dirac operator $D$, by thinking $D$ as proportional to ``the square root" of the Laplacian.
We have
\begin{equation}
\label{sacc0}
H(|D|)\propto\left\{\begin{array}{ll}
                    \!\!\!\!\!\!\! &|D|\,\,\text{for massless case}\,, \\[1ex]
               	\!\!\!\!\!\!\!&D^2\,\,\text{for massive case}\,. \\[1ex]
                    \end{array}
                    \right.
\end{equation}

We note that, the massless case can be referred to relativistic situations (i.e.\ photons or even gravitons), as well as to classical situations relative to the phonons. In these cases, the constant $c$ appearing in \eqref{sacc0000} corresponds to the speed of the light or the velocity of the sound in the considered medium, respectively.
Here, instead, the massive case is always referred to the non relativistic situation. 

As explained in \cite{DKS}, a spectral action is a function of the Dirac operator which has the form
\begin{equation*}
S=S(f,D,\La):=\tr\big(f(|D|/\La)\big)\,,
\end{equation*}
such that the l.h.s.\ is meaningful and, for the cut-off,
\begin{equation}
\label{cuof}
\La=\left\{\begin{array}{ll}
                    \!\!\!\!\!\!\! &1/\b\,\,\text{for massless case}\,, \\[1ex]
               	\!\!\!\!\!\!\!&1/\sqrt{\b}\,\,\text{for massive case}\,, \\[1ex]
                    \end{array}
                    \right.
\end{equation}
with $\b=\frac1{k_{\rm B}T}$ and $k_{\rm B}$ is the Boltzmann constant.

In \cite{DKS}, it was considered some particular classes of one-particle Hamiltonians for Fermi and Bose cases, and for $f$ it was taken extensive quantities as the entropy or the energy. Instead, in the present paper we consider the most general cases $q=[-1,1]$, hence including the physical case $q=0$ and $q=\pm1$. 

We note that no real classical massless particle exists in nature, except the cases describing the so-called {\it quasi-particles} (i.e.\ phonons). However, any $q$-particle, hence including fermions and bosons, behaves as a classical one in the low-density regime, that is when $z\approx0$. Therefore, all results listed in the present paper regarding massless particles for $q=0$, certainly can have some meaningful real application.

Summarising, we deal with more concrete situations involving free gases of massless particles (e.g.\  relativistic particles like photons and gravitons in the integer spin case, and families of neutrinos in the half-integer spin case provided the last are indeed massless, or also classical quasi-particles like phonons) and massive classical ones. Massive relativistic ones, like leptons and hadrons, and composite systems like atomic nuclei, are also treated in some details below. 

We point out that the spectral actions considered here naturally depend on the activity $z$, as well as $q$, satisfying the limitations
\begin{equation}
\label{constra}
\left\{\begin{array}{ll}
                    \!\!\!\!\!\!\! &0< z\,\,\text{if}\,\,q\in[-1,0]\,, \\[1ex]
             	\!\!\!\!\!\!\!&0< z<\frac{e^{\b\min\s(H)}}{q}\,\,\text{if}\,\,q\in(0,1]\,. \\[1ex]
                   \end{array}
                  \right.
\end{equation}

We always assume such limitations in the forthcoming analysis, and recall also that all functions considered in the present paper allow continuation in trivial case corresponding to $z=0$. We also suppose that the Hamiltonian $H$ is itself a function of the modulus $|D|$ of the Dirac operator $D$ according to \eqref{sacc0}.

The first quantity which can be used to build a spectral action, which arises directly from the grand partition function, is the
{\it Landau potential} (LP for short), called also {\it grand potential},
\begin{equation*}
\Omega_q(z):=-\frac{1}{\b}\ln Z_q(z)\,,\quad q\in[-1,1]\,.
\end{equation*} 

It is well known that, from a thermodynamic point of view, $\Omega_q$ can be expressed by $\Omega_q=-PV$, where $P$ is the pressure and $V$ is the volume occupied by a real physical system. Therefore, we obtain
\begin{equation*}
PV=-k_{\rm B}T\frac{\tr\big(\ln(I-zqe^{-\b H})\big)}{q}\,,\quad q\in[-1,1]\smallsetminus\{0\}\,,
\end{equation*} 
which reduces to $PV=k_{\rm B}Tz\,\tr\big(e^{-\b H}\big)$ for $q=0$, that is nothing but the average number $N$ by the well-known equation of states $PV=Nk_{\rm B}T$. 

Using \eqref{cuof} for the cut-off $\La$, 
the spectral action coming from the Landau potential assumes the form
$$
S(f_{\rm LP},D,\La):=-\b h\Om_q=\left\{\begin{array}{ll}
                    \!\!\!\!\!\!\! &-\frac{h}{q}\tr\big(\ln(I-zqe^{-\b H})\big)\,\,\text{if}\,\,q\in[-1,1]\smallsetminus\{0\}\,, \\[1ex]
             	\!\!\!\!\!\!\!&zh\,\tr\big(e^{-\b H}\big)\,\,\text{if}\,\,q=0\,. \\[1ex]
                   \end{array}                 
                   \right.
$$
Here, 
$$
f_{\rm LP}(|D|)=\left\{\begin{array}{ll}
                    \!\!\!\!\!\!\! &-\frac{h}{q}\ln\big(I-zqe^{-\b H(|D|)}\big)\,\,\text{if}\,\,q\in[-1,1]\smallsetminus\{0\}\,, \\[1ex]
             	\!\!\!\!\!\!\!&zh\,e^{-\b H(|D|)}\,\,\text{if}\,\,q=0\,,\\[1ex]
                   \end{array}                 
                   \right.
$$
and the inessential Planck constant $h\approx 6.626070040\times 10^{-34}\,Js$ is considered only for the sake of physical dimensionality.

Other natural spectral actions come from the, obviously extensive, quantities of mean entropy and energy of an infinitely extended system in thermal equilibrium.

For such a purpose, according to \eqref{pldibis} we start by defining the number function given by
\begin{equation}
\label{num0}
N_q(|D|):=\big(z^{-1}e^{\b H(|D|)}-q\big)^{-1}\,.
\end{equation}

Concerning the mean entropy (by using \eqref{0entc}, \eqref{000entc}, \eqref{entc9} and \eqref{pldi})
and energy (by using \eqref{plq1}, and again\eqref{pldi}), we have for the involved functions
$$
f_{\rm Ent}(|D|)\propto\left\{\begin{array}{ll}
                    \!\!\!\!\!\!\! &\frac{I+qN_q(|D|)}{q}\ln\big(I+qN_q(|D|)\big)-N_q(|D|)\ln N_q(|D|)\,\,\text{if}\,\,q\in[-1,1]\smallsetminus\{0\}\,, \\[1ex]
             	\!\!\!\!\!\!\!&N_0(|D|)\big(1-\ln N_0(|D|)\big)\,\,\text{if}\,\,q=0\,,\\[1ex]
                   \end{array}                 
                   \right.
$$
and
$$
f_{\rm En}(|D|)\propto N_q(|D|)H(|D|)\,.
$$
respectively.

\medskip

It is also natural to consider the spectral action associated merely to the average number of particles \eqref{num0}, again an extensive quantity, which seems to be the more flexible situation. Indeed, we have
$$
f_{\rm Number}(|D|):=hN_q(|D|)\,,
$$
and we have multiplied by $h$ for the sake of physical dimensionality as before. 

For the convenience of the reader, we report the following
\begin{Prop}
All the functions $f_\#(|D|)$ listed above are trace-class.
\end{Prop}
\begin{proof}
For the convenience of the reader, we provide the proof for the number spectral action, the unique we treat in detail in the present paper. The remaining ones follow similarly.

The Fermi/Bose cases are treated in \cite{BR}, which provide the proof after standard modifications. Indeed, the case $q=0$ is immediate because $N_q=z\tr e^{-\b H}$, and the cases $q\in[-1,0)$ reduces to the previous one since
$$
N_q=\tr\Big(\big(z^{-1}e^{\b H}+|q|\big)^{-1}\Big)\leq z\tr e^{-\b H}=N_0\,.
$$

Also the case $q\in(0,1]$ reduces to the first one because
$$
N_q=\tr\Big(\big(z^{-1}e^{\b H}-q\big)^{-1}\Big)\leq\frac{z\tr e^{-\b H}}{1-zqe^{-\b\min\s(H)}}=\frac{N_0}{1-zqe^{-\b\min\s(H)}}\,,
$$
after recalling that $z$ satisfies the constrain $z<\frac{e^{\b\min\s(H)}}{q}$ given in \eqref{constra}.
\end{proof}

We end the present section by noticing that, for non scalar particles (i.e.\ particles with a non zero spin), hence in the Bose/Fermi alternative, the one-particle Hibert space is a tensor product of a usual $L^2$-space with a finite dimensional multiplicity space describing the internal degrees of freedom associated to the spin. In all extensive quantities described above, this introduces an inessential multiplicative constant depending on the multiplicity which does not affect the minimal action principle. These simple considerations allow us to extend the analysis to quons, that is when $q\in(-1,0)\cup(0,1)$ for which such a multiplicity space is, on one hand inessential for statistical mechanical computations, and on the other hand completely unknown.

\section{The limit to the continuum}
\label{2mic2}

In the present section, we consider the spectral actions, after passing to the continuum, for infinitely extended systems living on $\br^d$, for which the involved Hamiltonians have continuum spectrum.  The expansion of the involved spectral actions for $\b\downarrow0$ consists only of the leading term, plus a constant term which takes into account possible condensation effects arising in the cases $q\in(0,1]$. In this situation, we put
\begin{equation}
\label{sacc0000}
H=\left\{\begin{array}{ll}
                    \!\!\!\!\!\!\! &cp\,\,\text{for massless case}\,, \\[1ex]
               	\!\!\!\!\!\!\!&\frac{p^2}{2m}\,\,\text{for massive case}\,. \\[1ex]
                    \end{array}
                    \right.
\end{equation}

Here, the Dirac operator (in Fourier transform) is the the momentum vector ${\bf p}$ whose modulus is 
$$
p:=\sqrt{{\bf p}\cdot{\bf p}}=\sqrt{\sum_{i=1}^n p_i^2}\,,
$$
$c\approx299792458\,ms^{-1}$ is the speed of the light, and $m$ is the mass of the involved particles. If the massless case concerns  phonons, $c$ should be considered as the velocity of the sound which, obviously, depends on the involved medium in which the waves are propagating.

It is customary (e.g.\ \cite{Hu}, and further \cite{FV, CF}) to perform the passage to the continuum as follows:
for a gas with $N$ particles, we simply make the replacement
\begin{equation}
\label{b4}
\sum_{\eps\in\s(H)}g(\eps)=\sum_{\{n(\eps)\mid\eps\in\s(H)\}}\rightarrow
V_d\int\frac{\di^d{\bf p}}{h^d}\,,
\end{equation}
$h$ being the Planck constant, and $V_d$ the $d$-dimensional volume of the physical system under consideration.

\medskip

To simplify, we only treat the spectral actions associated to the number-function \eqref{num0}, and split the matter in six cases listed below. We recall that $z>0$ for $q\in[-1,0]$, whereas $0<z<1/q$ for $q\in(0,1]$. We also denote by $\Om_d$, $d\geq1$, the $d$-dimensional ``solid" angle, that is the measure of the $(d-1)$-hypersurface of the hypersphere of radius 1 in $\br^d$, by noticing that $\Om_1=2$. Finally, for the involved spectral action $S$, we simply write $S_{\rm Number}(\b)$.
\medskip

\noindent
\textbf{Massless case, $\bm {q\in[-1,0)}$.} 
By using \eqref{b4}, \eqref{sacc0000} and \eqref{num0}, we obtain after an elementary change of variables,
$$
S_{\rm Number}(\b)=\left(\frac{V_d\Om_d}{|q|h^{d-1}c^d}\int_0^{+\infty}\di y\frac{y^{d-1}}{(z|q|)^{-1}e^y+1}\right)\frac1{\b^d}\,.
$$
\medskip

\noindent
\textbf{Massless case, $\bm{q=0}$.} By reasoning as above, we obtain
$$
S_{\rm Number}(\b)=\left(\frac{zV_d\Om_d}{h^{d-1}c^d}\int_0^{+\infty}\di y\,y^{d-1}e^{-y}\right)\frac1{\b^d}\,.
$$
\noindent
\textbf{Massive case, $\bm{q\in[-1,0)}$.}
As before, we obtain
$$
S_{\rm Number}(\b)=\left(\frac{V_d\Om_d(2m)^{d/2}}{|q|h^{d-1}}\int_0^{+\infty}\di y\frac{y^{d-1}}{(z|q|)^{-1}e^{y^2}+1}\right)\frac1{\b^{d/2}}\,.
$$
\medskip

\noindent
\textbf{Massive case, $\bm{q=0}$.} We have
$$
S_{\rm Number}(\b)=\left(\frac{zV_d\Om_d(2m)^{d/2}}{h^{d-1}}\int_0^{+\infty}\di y\,y^{d-1}e^{-y^2}\right)\frac1{\b^{d/2}}\,.
$$

Due to the phenomenon of condensation of particles in the fundamental state, known as {\it Bose-Einstein condensation} in the case $q=1$, see \cite{BR, Hu, AF5, FV, CF, Z}, the cases $q\in(0,1]$ deserve a more refined analysis.

It is well know that the condensation takes place if and only if the {\it critical density}
\begin{align*}
\r_{\rm c}(q):=\lim_{z\uparrow(1/q)}(N/V_d)\propto&\lim_{z\uparrow(1/q)}\int_0^{+\infty}\di y\frac{y^{d-1}}{(zq)^{-1}e^{y^\a}-1}\\
=&\int_0^{+\infty}\di y\frac{y^{d-1}}{e^{y^\a}-1}\propto\r_{\rm c}(1)<+\infty\,.
\end{align*}

Here, $
\a=\left\{\begin{array}{ll}
                    \!\!\!\!\!\!\! &1\,\,\text{for massless case}\,, \\[1ex]
               	\!\!\!\!\!\!\!&2\,\,\text{for massive case}\,, \\[1ex]
                    \end{array}
                    \right.$ and the above integral is convergent if and only if $\a-d+1<1$. Therefore, the condensation takes place if and only if  
                    $\left\{\begin{array}{ll}
                    \!\!\!\!\!\!\! &d\geq2\,\,\text{for massless case}\,, \\[1ex]
               	\!\!\!\!\!\!\!&d\geq3\,\,\text{for massive case}\,. \\[1ex]
                    \end{array}
                    \right.$  
It is then customary to separate the contribution corresponding to ground level $\eps=0$ from the remaining ones, see e.g.\ \cite{Hu}, and \cite{CF}, Section 5. Therefore, we set
$$
a\,\,\left\{\begin{array}{ll}
                    \!\!\!\!\!\!\! &=0\,\,\text{for massless case when $d=1$ and massive case when $d=1,2$}\,, \\[1ex]
               	\!\!\!\!\!\!\!&\geq0\,\,\text{for massless case when $d\geq2$ and massive case when $d\geq3$}\, \\[1ex]
                    \end{array}
                    \right.
$$
where $a$, describing the strength of the possible amount of condensate, is a dimensionless number.

Indeed, in the condensation regime, we deal with a multi-phase situation, and thus the portion of the condensate, possibly also vanishing, can vary according to many a-priori constrains, such as the boundary conditions used to reach the thermodynamic limit, see e.g.\ \cite{BR}, Section 5.2.5. Therefore, the coefficient $a$ will take into account of the possible presence of the condensate. 
We also note that, in inhomogeneous systems the appearance of the condensate is not only connected with the finiteness of the critical density, see \cite{F25, F26, F27}.

We then have the following situations. 
\medskip

\noindent
\textbf{Massless case, $\bm{q\in(0,1]}$.}
By using \eqref{b4}, \eqref{sacc0000} and \eqref{num0}, we obtain after an elementary change of variables,
\begin{equation}
\label{1asdt1}
S_{\rm Number}(\b)=\left(\frac{V_d\Om_d}{qh^{d-1}c^d}\int_0^{+\infty}\di y\frac{y^{d-1}}{(zq)^{-1}e^y-1}\right)\frac1{\b^d}
+a\frac{hz}{1-zq}\,.
\end{equation}
\medskip

\noindent
\textbf{Massive case, $\bm{q\in(0,1]}$.}
Again by using \eqref{b4}, \eqref{sacc0000} and \eqref{num0}, we obtain after an elementary change of variables,
\begin{equation}
\label{1asdt0}
S_{\rm Number}(\b)=\left(\frac{V_d\Om_d(2m)^{d/2}}{qh^{d-1}}\int_0^{+\infty}\di y\frac{y^{d-1}}{(zq)^{-1}e^{y^2}-1}\right)\frac1{\b^{d/2}}
+a\frac{hz}{1-zq}\,.
\end{equation}

It should be noted that,
\begin{itemize}
\item[(i)] if $\b\sim0$, the right-hand addendum in the last two cases, describing the possible presence of the condensate, is obviously negligible;
\item[(ii)] the meaning of the condensation of photons should be understood as explained in \cite{KSVW}, see also \cite{FV}, Section 6.
\end{itemize}

\section{Finite volume theory}
\label{fvte}

We investigate the finite volume theory for arbitrary dimension $d=1,2,\dots$, and $q\in[-1,1]$. For such a purpose, we use the periodic conditions on the $d$-interval $[0,L]^d$. We also recall that a quantity $F=o\big(\eps^\infty\big)$ as $\eps\to\#$, simply means
$$
\lim_{\eps\to\#}F(\eps)/\eps^n=0\,,\quad \text{for all}\,\,n\in\bn\,.
$$

The momentum operator ${\bf p}$ of a particle living in the manifold $[0,L]^d\sim\bt^d$, is given by
$$
p_i:=-\imath\hbar\frac{\partial\,\,}{\partial x_i}\,,\quad i=1,2,\dots, d\,,
$$           
acting on $L^2([0,L]^d,\di^d x)$. 

By passing in momentum space through the Fourier transform, and using the wave-vector ${\bf k}\in\bz^d$, we easily compute by imposing the periodic boundary conditions,
$$
{\bf p}=\frac{h}{L}{\bf k}\,,\quad {\bf k}\in\bz^d\,,
$$
and thus
$$
p^\a=\Big(\frac{h}{L}\big({\bf k} \cdot {\bf k}\big)^{1/2}\Big)^\a\,,\quad {\bf k}\in\bz^d
$$
where, as before, $\a=1$ in the massless case, and $\a=2$ in the massive case.

\medskip

Concerning the asymptotics, we treat first the simpler massive case.
\begin{Prop}
\label{mssv}
For the spectral action $S_{\rm Number}(\b)$ in the massive case,
$$
S_{\rm Number}(\b)=\left(\frac{L^d\Om_d(2m)^{d/2}}{h^{d-1}}\int_0^{+\infty}\di y\frac{y^{d-1}}{z^{-1}e^{y^2}-q}\right)\frac1{\b^{d/2}}+o\big(\b^\infty\big)\,,\,\,\,\text{as}\, \,\,\b\downarrow0\,,
$$
where, for the activity, $0<z$ for $q\in[-1,0]$, and $0< z<1/q$ for $q\in(0,1]$.
\end{Prop}
\begin{proof}
The proof directly follows by \cite{EI}, Proposition 2.27 (see also \cite{EGV}, Lemma 2.10). Indeed, with $\eps:=\sqrt{\b}$,
$$
S_{\rm Number}(\eps^2)\equiv h\sum_{{\bf k}\in\bz^d}\frac1{z^{-1}e^{\frac{h^2}{2mL^2}(\eps k)^2}-q}=h\int_{\br^d}\frac{\di^d{\bf k}}{z^{-1}e^{\frac{h^2}{2mL^2}(\eps k)^2}-q}
+o\big(\eps^{\infty}\big)\,.
$$

The assertion now follows after the elementary change of variable $y=\frac{h\eps}{(2m)^{1/2}L}k$, and integrating on the solid angles.
\end{proof}

Due to the particular form of the involved Hamiltonian, the massless case presents less regularity properties. 
\begin{Prop}
\label{lere}
For the spectral action $S_{\rm Number}(\b)$ in the massless case,
$$
S_{\rm Number}(\b)=\left(\frac{L^d\Om_d}{c^dh^{d-1}}\int_0^{+\infty}\di y\frac{y^{d-1}}{z^{-1}e^{y}-q}\right)\frac1{\b^{d}}+o\big(\b^{-d}\big)\,,\,\,\,\text{as}\, \,\,\b\downarrow0\,,
$$
where, for the activity, $0<z$ for $q\in[-1,0]$, and $0<z<1/q$ for $q\in(0,1]$.
\end{Prop}
\begin{proof}
We start to prove the assertion for $d=1$. Indeed, in such a situation, 
\begin{equation}
\label{asdt000}
S_{\rm Number}(\b)=2h\sum_{k=1}^{+\infty}\frac1{z^{-1}e^{\frac{\b ch}{L}k}-q}+\frac{hz}{1-zq}\,.
\end{equation}
In addition,
\begin{align*}
h\sum_{k=1}^{+\infty}\frac1{z^{-1}e^{\frac{\b ch}{L}k}-q}
\leq&\left(\frac{L}{c}\int_0^{+\infty}\frac{\di y}{z^{-1}e^{y}-q}\right)\frac1{\b}\\
\leq&h\sum_{k=1}^{+\infty}\frac1{z^{-1}e^{\frac{\b ch}{L}k}-q}+\frac{hz}{1-zq}\,.
\end{align*}
Therefore,
$$
0\leq\left(\frac{L}{c}\int_0^{+\infty}\di y\frac{\di y}{z^{-1}e^{y}-q}\right)\frac1{\b}-h\sum_{k=1}^{+\infty}\frac1{z^{-1}e^{\frac{\b ch}{L}k}-q}
\leq\frac{hz}{1-zq}\,,
$$
and thus
\begin{align*}
0\leq\left[\left(\frac{L}{c}\int_0^{+\infty}\di y\frac{\di y}{z^{-1}e^{y}-q}\right)\frac1{\b}-h\sum_{k=1}^{+\infty}\frac1{z^{-1}e^{\frac{\b ch}{L}k}-q}\right]\bigg/(1/\b)
\leq\frac{\b hz}{1-zq}\to0\,,
\end{align*}
as $\b\downarrow0$. We then obtain
$$
h\sum_{k\in\bz}^{+\infty}\frac1{z^{-1}e^{\frac{\b ch}{L}|k|}-q}
=\left(\frac{2L}{c}\int_{0}^{+\infty}\frac{\di y}{z^{-1}e^{y}-q}\right)\frac1{\b}+o\big(1/\b\big)\equiv O\big(1/\b\big)\,.
$$

Suppose that the assertion holds true for each integer $0\leq s\leq d-1$, and compute as before. First we note that
\begin{equation}
\label{asdt0}
S_{\rm Number}(\b)=2^d h\sum_{k_1,\dots,k_d\geq1}\frac1{z^{-1}e^{\frac{\b ch}{L}k}-q}+O\big(1/\b^{d-1}\big)\,,
\end{equation}
where the factor $2^d$ appears after partitioning the whole sum in partial sums on the $2^d$ $d$-ants (which are $2$ half-lines when $d=1$, $4$ quadrants when $d=2$, $8$ octants when $d=3$, and so on).

On the other hand,
\begin{equation}
\label{asdt}
\begin{split}
0\leq&\left(\frac{L^d\Om_d}{2^dc^dh^{d-1}}\int_0^{+\infty}\di y\frac{y^{d-1}}{z^{-1}e^{y}-q}\right)
\frac1{\b^{d}}-h\sum_{k_1,\dots,k_d\geq1}\frac1{z^{-1}e^{\frac{\b ch}{L}k}-q}\\
\leq&O\big(1/\b^{d-1}\big)\,.
\end{split}
\end{equation}

Collecting together \eqref{asdt0} and \eqref{asdt}, we obtain the assertion.
\end{proof}

\section{One dimensional case.}

The asymptotics of the number-function for the massive case when $d=1$ and $q=0$ can be explicitly computed by using an appropriate Jacobi theta function, see e.g.\ \cite{WW}, pag. 464. Indeed
(e.g.\  \cite{EI}, Example 2.41), we get
\begin{align*}
S_{\rm Number}(\b)=&\sqrt{2m\pi}L\frac1{\sqrt{\b}}+2\sqrt{2m\pi}L\frac1{\sqrt{\b}}
\sum_{n=1}^{+\infty}e^{-\frac{2mL^2\pi^2n^2}{h^2\b}}\\
=&\sqrt{2m\pi}L\frac1{\sqrt{\b}}+o\big(\b^\infty\big)\,,\,\,\,\text{as}\, \b\downarrow0\,.
\end{align*}

In the massless case when $d=1$, we can provide the following more refined result concerning the asymptotics. The first case corresponds to $q=0$ for which it can be also explicitly computed as in the massive case above.

Indeed, we have for sufficiently small $\b$,
\begin{equation}
\label{anzze}
\begin{split}
S_{\rm Number}(\b)=&zh\sum_{k\in\bz}e^{-\frac{ch\b}{L}|k|}=zh\frac{e^{\frac{ch\b}{L}}+1}{e^{\frac{ch\b}{L}}-1}\\
=&\frac{2zL}{c}\frac1{\b}+2zh\sum_{n=1}^{+\infty}\Big(\frac{ch}{L}\Big)^{2n-1}\frac{B_{2n}}{(2n)!}\b^{2n-1}\,,
\end{split}
\end{equation}
where $B_{2n}$ are the Bernoulli numbers.
We now extend the analysis to all Fermi-like situations $q\in[-1,0)$.
\begin{Prop}
\label{mls}
For the spectral action $S_{\rm Number}(\b)$ in the massless case when $d=1$ and $q\in[-1,0]$,
\begin{equation}
\label{berrie}
\begin{split}
S_{\rm Number}(\b)=&\left(\frac{2L}{c}\int_0^{+\infty}\frac{\di y}{z^{-1}e^{y}-q}\right)\frac1{\b}+\frac{hz}{1-zq}\\
+&2h\sum_{n=0}^{+\infty}\frac{\z(-n)c_n}{n!}\b^n+o\big(\b^{\infty}\big)\,,\,\,\,\text{as}\, \b\downarrow0
\end{split}
\end{equation}
where, for the activity, $0< z$ for $q\in[-1,0]$, $\z$ is the Riemann zeta function, and finally 
$c_n=\frac{\di^n\,\,}{\di x^n}\Big(\frac1{z^{-1}e^{\frac{ch}{L}x}-q}\Big)\Big|_{x=0}$.
\end{Prop}
\begin{proof}
Since the case $q=0$ is explicitly described above, we can limit the analysis to $q<0$. In such a situation, $f(x)=\frac1{z^{-1}e^{\frac{ch}{L}x}-q}$, $x\in\br$, is infinitely often differentiable with $f(x)=O(1)$ and $f^{(n)}(x)=o\big((1/x)^\infty\big)$, $n\geq1$, for $|x|\to\infty$ and thus, with $\g=0$, 
$f\in\ck_\g(\br)\subset\ck(\br)$ (for the definition of the class of smooth functions $\ck$ see 
\cite{E, EGV}). In addition, $f\in L^1([0,+\infty),\di x)$, hence 
the assertion follows from \cite{EGV}, Lemma 2.11.
\end{proof}

One is tempted to extend the above mentioned result to the remaining cases, conjecturing that 
\eqref{berrie} would hold true also for $q\in(0,1]$.

Unfortunately, Lemma 2.11 in \cite{EGV} cannot be applied to \eqref{berrie} because it is unclear how to find a function $f\in\ck(\br)$ such that
$f\lceil_{[0,+\infty)}=\frac1{z^{-1}e^{\frac{ch}{L}x}-q}$. Such a mentioned lemma cannot be applied also to the function $e^{-x}$ involved in the case $q=0$, since it does not belong to the class $\ck$.
Yet, the expansion of $S_{\rm Number}$ is explicitly computed in \eqref{anzze}.

To be more precise, given a smooth function $g:[0,+\infty)\to[0,+\infty)$, only supported on the right half-line, such that 
$\sum_{k=1}^{+\infty}g(\eps k)<+\infty$ in a right neighborhhood $\eps\in(0,\d)$ of zero, the analogous of \eqref{berrie} would be
\begin{equation}
\label{berrie0}
\sum_{k=1}^{+\infty}g(\eps k)\sim\frac1{\eps}\int_0^{+\infty}g(x)\di x+\sum_{n=0}^{+\infty}\frac{\z(-n)g^{(n)}(0^+)}{n!}\eps^n\,,
\,\,\text{as}\, \eps\downarrow0\,.
\end{equation}

For such a kind of functions $g$, we note that the necessary conditions under which the r.h.s.\ of \eqref{berrie0} is meaningful are the following ones. Firstly, we should have
\begin{itemize}
\item[(a)] $g\in L^1\big([0,+\infty)\big)$.
\end{itemize}
By taking into account the asymptotics of the Bernoulli numbers (cf.\ \cite{Lee}) in (1.1) of \cite{A} and the Stirling formula, we compute
$$
\limsup_n\bigg(\frac{|B_{2n}|}{(2n)!}|g^{(2n)}(0)|\bigg)^\frac1{2n}\sim\frac1{2\pi}\limsup_n\Big(|g^{(2n)}(0)|^\frac1{2n}\Big)\,.
$$
Therefore, secondly, for the convergence of the series in the r.h.s. of \eqref{berrie0} in some neighborhhood of zero,
we must have
\begin{itemize}
\item[(b)] $\limsup_n\Big(|g^{(2n)}(0)|^\frac1{2n}\Big)<+\infty$.
\end{itemize}

It would be then meaningful to address the, still open, relevant question of whether the above conditions (a) and (b) on $g$ are also sufficient to assure the validity of \eqref{berrie0}. In this respect, we want to mention \cite{BE} in which such an analysis is carried out for some very particular examples appearing in Chapter 15 of Ramanujan's second notebook without any proof. 

It seems that, however, the involved series does converge also in the Bose-like cases for which 
$0< z<1/q$ when $q\in(0,1]$. Indeed, for the function $g(x)=\frac1{z^{-1}e^{\frac{ch}{L}x}-q}$, (a) is easily satisfied, and (b) above seems also to be satisfied as depicted in the figure below.
\begin{figure}[h!]
\begin{center}
  \includegraphics[width=10cm]{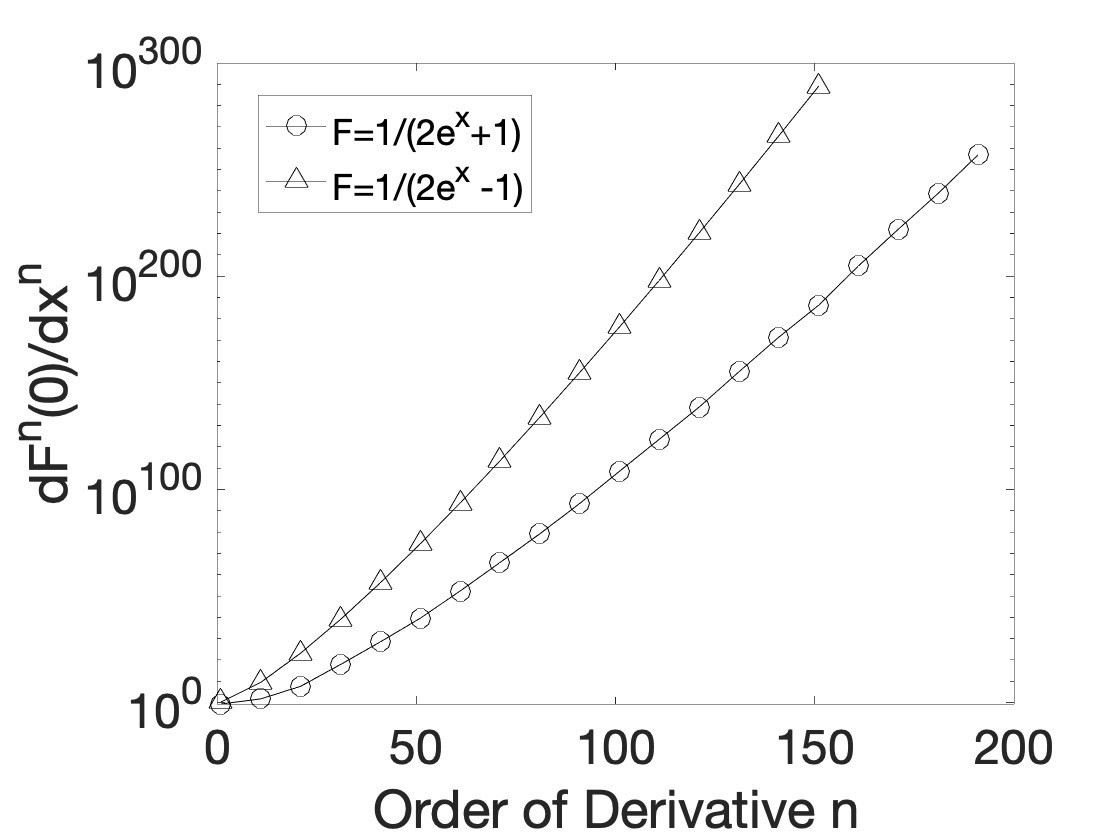}\\
  \caption{$n$-derivatives in the origin in logarithmic scale}
\end{center}
\end{figure}
\medskip

\section{The sphere $\bm{\bs^3}$.}
To provide a manageable three-dimensional example, we consider massless and massive quons living in the sphere 
$$
\bs_R^3:=\{(x,y,z,t)\in\br^4\mid x^2+y^2+z^2+t^2=R^2\}
$$
with radius $R>0$. We point out that, for the sphere $\bs^d\subset\br^{d+1}$, the general odd cases $d=2l+1$, $l>1$, can be similarly handled, whereas the even cases are more simple. We recall that $\hbar=\frac{h}{2\pi}$.

We start by recalling that the spectrum of the Dirac operator of $\bs_1^3\equiv\bs^3$ is 
$\{\pm(n+1/2)\mid n\in\bn\}$ with multiplicity $2\binom{n+1}{n}$, see e.g\ \cite{G}, Theorem 2.1.3. We first treat the massive case.
\begin{Prop}
For the massive $q$-particles living in the three-dimensional sphere $\bs_R^3$ in $\br^4$ of radius $R$, we get
\begin{align*}
S_{\rm Number}(\b)=&\left(4\pi\frac{\big(\sqrt{2m}R\big)^3}{\hbar^2}\int_0^{+\infty}\di y\frac{y^2}{z^{-1}e^{y^2}-q}\right)\frac1{\b^{3/2}}\\
-&\left(\pi\sqrt{2m}R\int_0^{+\infty}\frac{\di y}{z^{-1}e^{y^2}-q}\right)\frac1{\b^{1/2}}
+o\big(\b^\infty\big)\,,\,\,\,\text{as}\, \b\downarrow0\,.
\end{align*}
\end{Prop}
\begin{proof}
The proof is provided in \cite{CC1}, Section 2.2. Indeed, by using the spectral resolution of the identity of the Dirac operator on the sphere described above, we get
\begin{equation*}
S_{\rm Number}(\b)=\sum_{k\in\bz}k(k+1)f\Big(\frac{k+1/2}{\La}\Big)
\end{equation*}
where, in our situation, $\La=1/\sqrt{\b}$ and
$$
f(x)=\frac{h}{z^{-1}e^{\frac{\hbar^2}{2mR^2}x^2}-q}\in\cs(\br)\,.
$$

Therefore, reasoning as in \cite{CC1}, we obtain the assertion after an elementary change of variable.
\end{proof}

The massless case is much more complicated than the massive one because the involved function
$\frac{h}{z^{-1}e^{\frac{c\hbar}{R}|x|}-q}$    
is not smooth. However, we are still able to provide the asymptotic expansion of $S_{\rm Number}$, at least for the Fermi-like and Boltzmann particles.
\begin{Prop}
\label{masswr}
For the massless $q$-particles, $q\in[-1,0]$, living in the three-dimensional sphere $\bs_R^3$ in $\br^4$ of radius $R$, we get
\begin{equation}
\label{esvarg101}
\begin{split}
S_{\rm Number}(\b)=&\left(\frac{4\pi R^3}{c^3\hbar^2}\int_0^{+\infty}\di y\frac{y^2}{z^{-1}e^{y}-q}\right)\frac1{\b^{3}}\\
+&\sum_{n=-2}^{+\infty}b_n\b^n
+o\big(\b^\infty\big)\,,\,\,\,\text{as}\, \b\downarrow0\,.
\end{split}
\end{equation}
\end{Prop}
We note that the coefficients $b_n$ can be explicitly determined as explained in the forthcoming proof.
\begin{proof}
With $\eps=\b=1/\La$, and considering the spectral behaviour of the Dirac operator on the sphere $\bs^3$ described above, we compute
\begin{align*}
S_{\rm Number}(1/\La)=&\sum_{k=1}^{+\infty}k(k+1)f\Big(\frac{2k+1}{2\La}\Big)\\
=&\sum_{k=3}^{+\infty}\left[\Big(\frac{k-1}2\Big)^2+\frac{k-1}2\right]f\Big(\frac{k}{2\La}\Big)
-\sum_{k=2}^{+\infty}\left[\Big(k-\frac12\Big)^2+k-\frac12\right]f\Big(\frac{k}{\La}\Big)\\
=&\sum_{k=1}^{+\infty}\left[\Big(\frac{k-1}2\Big)^2+\frac{k-1}2\right]f\Big(\frac{k}{2\La}\Big)
-\sum_{k=1}^{+\infty}\left[\Big(k-\frac12\Big)^2+k-\frac12\right]f\Big(\frac{k}{\La}\Big)\\
=&\sum_{k=1}^{+\infty}\Big(\frac{k^2}4-\frac14\Big)f\Big(\frac{k}{2\La}\Big)
-\sum_{k=1}^{+\infty}\Big(k^2-\frac14\Big)f\Big(\frac{k}{\La}\Big)\\
=&\La^2\sum_{k=1}^{+\infty}\left[g\Big(\frac1{2\La}k\Big)-g\Big(\frac1{\La}k\Big)\right]
+\frac14\sum_{k=1}^{+\infty}\left[f\Big(\frac1{\La}k\Big)-f\Big(\frac1{2\La}k\Big)\right]\,.
\end{align*}
Therefore,
\begin{equation}
\label{esvarg1}
S_{\rm Number}(\eps)=\frac1{\eps^2}\sum_{k=1}^{+\infty}\left[g\Big(\frac{\eps}2 k\Big)-g(\eps k)\right]
+\frac14\sum_{k=1}^{+\infty}\left[f(\eps k)-f\Big(\frac{\eps}2 k\Big)\right]\,.
\end{equation}
Since the case $q=0$ can be explicitly handled as at the beginning of the previous section,
we reduce the matter to $q<0$. 

For such a purpose, we first observe that $f(x):=\frac{2h}{z^{-1}e^{\frac{c\hbar}{R}x}-q}$ belongs to $\ck_0(\br)$, $g(x):=x^2\frac{2h}{z^{-1}e^{\frac{c\hbar}{R}x}-q}$ belongs to 
$\ck_2(\g)$, and thus both $f$ and $g$ belong to $\ck(\br)$. In addition, $f,g\in L^1([0,+\infty),\di x)$.
Hence, for the remaining cases $q\in[-1,0)$, we can apply
Lemma 2.11 in \cite{EGV} to the four series in  \eqref{esvarg1}, allowing us to explicitly compute the searched expansion for $S_{\rm Number}$. 

Concerning the leading term, which comes from the first two series, we obtain after the elementary change of variables $y=\frac{c\hbar}{R}x$,
$$
S_{\rm Number}(\eps)\sim\frac{2h}{\eps^3}\int_0^{+\infty}\di x\frac{x^2}{z^{-1}e^{\frac{c\hbar}{R}x}-q}
=\left(\frac{4\pi R^3}{c^3\hbar^2}\int_0^{+\infty}\di y\frac{y^2}{z^{-1}e^{y}-q}\right)\frac1{\eps^3}\,.
$$
\end{proof}

As explained in the previous section, we can conjecture that \eqref{esvarg101} might hold true also for the remaining cases $q\in(0,1]$.

\section{Examples and remarks}
\label{rmsex}

This section is devoted to some remarks relative to the topic of the present paper. We also consider in some detail the relativistic massive case.
\medskip

\noindent
\textbf{Comparison with the infinitely extended systems.}
Firstly, we want to note that the condensation phenomena cannot appear, at least in a purely mathematical investigation, if the bottom of the  spectrum of the one-particle Hamiltonian 
$\min\s(H)$ is an eigenvalue, see e.g.\ \cite{BR}, Section 5.2.2.
This certainly happens in a finite-volume theory, or also for an infinitely extended system made of noninteracting harmonic oscillators, that is when the one-particle Hamiltonian is given by $H=\frac{p^2}{2m}+K^2 q^2$, $K$ being the Hooke elastic constant.
Therefore, the constant term in \eqref{asdt000} has a completely different meaning w.r.t. those in \eqref{1asdt1} and \eqref{1asdt0}. 

\medskip

The second thing we want to point out is that the leading terms of the spectral actions of the expansions w.r.t.\ $1/\b=k_{\rm B}T$ of the finite volume theories, coincide with the corresponding ones obtained after the passage to the continuum (compare the computations in Sections \ref{fvte} 
and \ref{2mic2}). We also note that, in the latter situation, there is only one term, the leading one, in the expansion of the spectral actions. In other words,
all possible other terms automatically disappear after the passage to the continuum.
The possible appearance of a constant term in the infinite volume theory is connected only with the condensation phenomena as explained above. It would be of interest to give a reasonable physical explanation of these facts.
\medskip

\noindent
\textbf{On the asymptotics in the multi-dimensional massless case.}
Considering together Propositions \ref{lere} and \ref{mls}, one would conjecture that
\begin{equation}
\label{anzaf}
\begin{split}
S_{\rm Number}(\b)=&\sum_{s=0}^{d-1}\left(\binom{d}{s}\frac{L^{d-s}\Om_{d-s}}{c^{d-s}h^{d-s-1}}\int_0^{+\infty}\di y\frac{y^{d-s-1}}{z^{-1}e^{y}-q}\right)\frac1{\b^{d-s}}+\frac{hz}{1-zq}\\
+&2dh\sum_{n=0}^{+\infty}\frac{\z(-n)c_n}{n!}\b^n+o\big(\b^{\infty}\big)\,,\,\,\,\text{as}\, \b\downarrow0\,.
\end{split}
\end{equation}
Here, the above ansatz comes by partitioning the underlying lattice $\bz^d$ in suitable subsets as described in the proof of Proposition \ref{lere}.

The proof of \eqref{anzaf}, which describes the expansion of the spectral actions in the finite volume theories when $d>1$, would follow if one proved the multi-dimensional analogue of \cite{EGV}, Lemma 2.11.  In this multi-dimensional context, we note that the expansion \eqref{esvarg101} for the 3-dimensional sphere indeed has a form similar to \eqref{anzaf}.
\medskip

\noindent
\textbf{Relativistic massive Hamiltonian.}
We now briefly discuss the asymptotics associated to the massive relativistic Hamiltonian given by
$$
H=c\sqrt{p^2+m^2c^2}
$$
where, as usual, $c$ is the speed of the light and $m>0$ the mass at rest of the involved particles. We reduce the matter to a noninteracting gas of relativistic particles in a finite volume in $\br^d$ for the case $q=0$. To simplify, we also put $c=1=m$ and $z=1$ for the activity, all other remaining cases (i.e.\ for example for Bose/Fermi cases when $q\pm1$) can be handled similarly. 

In such a situation, we must investigate the asymptotics of a series of the form
\begin{equation}
\label{sumto}
\sum_{{\bf k}\in\bz^d}e^{-\b\sqrt{k^2+1}}=\sum_{{\bf k}\in\bz^d}e^{-\sqrt{(\b k)^2+\b^2}}
=\sum_{{\bf k}\in\bz^d}f_\b(\b k)\,.
\end{equation}

Since no result involving the detailed asymptotics is available for series like those described above, we limit the analysis to extract the leading term. As we have seen before, we can argue that this leading term directly provides the ``limit to the continuum". We also note that some condensation effects (for the cases $q>0$), neglected in this analysis, may also take place in the relativistic situation, see e.g.\ \cite{Z}.

As usual, by approximating a series as in \eqref{sumto} with an integral, we must handle 
$\left(\int_0^{+\infty}\di x\,x^{d-1}e^{-\sqrt{x^2+\b^2}}\right)\frac1{\b^d}=\frac{I(\b)}{\b^d}$, 
where $I(\b):=\int_0^{+\infty}\di x\,x^{d-1}e^{-\sqrt{x^2+\b^2}}$. 
\begin{Prop}
For the series in \eqref{sumto}, we get
$$
\sum_{{\bf k}\in\bz^d}e^{-\b\sqrt{k^2+1}}=\Big(\Om_d\int_0^{+\infty}\di x\,x^{d-1}e^{-x}\Big)\frac1{\b^d}+o(1/\b^d)\,.
$$
\end{Prop}
\begin{proof}
By following the same lines in Proposition \ref{lere}, we get
$$
\sum_{{\bf k}\in\bz^d}e^{-\b\sqrt{k^2+1}}=\Big(\Om_d\int_0^{+\infty}\di x\,x^{d-1}e^{-\sqrt{x^2+\b^2}}\Big)\frac1{\b^d}+o(1/\b^d)\,.
$$

Since $\sqrt{x^2+\b^2}$ decreases as $\b\downarrow0$, a simple application of the monotone convergence theorem leads to
$$
\int_0^{+\infty}\di x\,x^{d-1}e^{-\sqrt{x^2+\b^2}}-\int_0^{+\infty}\di x\,x^{d-1}e^{-x}=o(1)\,,
$$
and thus the assertion follows.
\end{proof}

\section*{Acknowledgement} 

The authors acknowledge MIUR Excellence Department Project awarded to the
Department of Mathematics, University of Rome ``Tor Vergata'', CUP
E83C18000100006, and Italian INdAM-GNAMPA. 
The first-named author is partially supported by MIUR-FARE R16X5RB55W QUEST-NET. The authors are also grateful to S. Marullo for a careful numerical computation relative to the number spectral actions, reported in the picture in the present paper.

\end{document}